\documentclass{llncs}

\usepackage{amssymb}
\usepackage{amsmath}
\usepackage{textcomp}
\usepackage{array}
 \usepackage[utf8]{inputenc}
\usepackage{lineno}
\usepackage{color}
\usepackage{comment}

\usepackage{graphicx}

\newcommand{\ket}[1]{|#1\rangle}

\newcommand{\Endproof}{\hfill$\Box$\\}

\title{Very narrow quantum OBDDs and width hierarchies for classical OBDDs}
\author{Farid Ablayev$^{1,}$\thanks{Partially supported by RFBR Grant 14-01-06036.} \and Aida Gainutdinova$^{1,\star}$ \and Kamil Khadiev$^{1,}$\thanks{Partially supported by RFBR Grant 14-07-00557.} \and Abuzer Yakary{\i}lmaz$^{2,3,}$\thanks{Work done in part while visiting the Kazan Federal University. Partially supported by CAPES, ERC Advanced Grant MQC, and FP7 FET project QALGO.}}
\institute{Kazan Federal University, Kazan, Russia       
       \and
      University of Latvia, Faculty of Computing, Raina bulv. 19 Riga, LV-1586, Latvia 
      \and
      National Laboratory for Scientific Computing, Petr\'{o}polis, RJ, 25651-075, Brazil
      \\ \email{fablayev@gmail.com,aida.ksu@gmail.com,kamilhadi@gmail.com,abuzer@lncc.br} 
}

\begin{document}

\maketitle

\begin{abstract}
We present several results on comparative complexity  for different variants of OBDD models.
 \begin{itemize}
   \item We present some results on comparative complexity of classical and quantum OBDDs. We consider a partial function depending on parameter $k$ such that for any $k>0$  this function is computed by an exact quantum OBDD of width $2$ but any classical OBDD (deterministic or 
  stable bounded error probabilistic) needs width $2^{k+1}$.
   \item We consider quantum and classical nondeterminism. We show that quantum nondeterminism can be more efficient than classical one. In particular,  an explicit function  is presented which is computed by a quantum nondeterministic OBDD with constant width but any classical nondeterministic OBDD for this function needs non-constant width.
   \item We also present new hierarchies on widths of deterministic and nondeterministic OBDDs. We focus both on small and large widths.
 \end{itemize}
\end{abstract}




\section{Introduction}

Branching programs are one of the well known models of computation. These models have been shown useful in a variety of domains, such as hardware verification, model checking, and other CAD applications (see for example the book by Wegener \cite{Weg00}). It is known that the class of Boolean  functions computed by polynomial size branching programs are coincide with the class of functions computed by non-uniform log-space machines. Moreover branching program is a convenient model for  considering different (natural)  restrictive  variants and different complexity measures such as size (number of inner nodes), length, and width.

One of important restrictive branching programs is oblivious read once branching programs, also known in applied computer science as Ordered Binary Decision Diagrams (OBDD) \cite{Weg00}.  Since the length of an OBDD is at most linear (in the length of the input), the main complexity measure is ``width''.

OBDDs also can be seen as nonuniform automata (see for example \cite{ag05}).  During the last decades different variants of OBDDs were considered, i.e. deterministic, nondeterministic, probabilistic, and quantum, and
many results have been proved on comparative power of deterministic, nondeterministic, and randomized OBDDs \cite{Weg00}. For example, Ablayev and Karpinski \cite{ak96} presented the first function that is polynomially easy for  randomized and exponentially hard for deterministic and even nondeterministic OBDDs. More specifically, it was proven that OBDD variants of  $ \mathsf{coRP} $ and $\mathsf{NP}$ are different.

In the last decade quantum model of OBDD came into play \cite{agk01},\cite{nhk00},\cite{SS05}.  It was proven that quantum OBDDs can be exponentially cheaper than classical ones and it was shown that this bound is tight \cite{agkmp05}.

In this paper we present the first results on comparative complexity  for classical and quantum OBDDs computing partial functions. Then, we focus on the  width complexity of deterministic and nondeterministic OBDDs, which have been investigated in different papers (see for more information and citations \cite{hs00}, \cite{hs03}). Here we present very strict hierarchies for the classes of Boolean functions computed by deterministic and nondeterministic OBDDs.

The paper is organized as follows. The Section 2 contains the definitions and notations used in the paper. In Section 3, we compare classical and exact quantum OBDDs. We consider a partial function depending on parameter $k$ such that, for any $k>0$, this function is computed by an exact quantum OBDD of width $2$ but deterministic or bounded error probabilistic OBDDs need width $2^{k+1}$. Also it is easy to show that nondeterministic OBDD needs width $k+1$. In Section 4, we consider quantum and classical nondeterminism. We show that quantum nondeterministic OBDDs can be more efficient than their classical counterparts. We present  an explicit function which is computed by a quantum nondeterministic OBDD with constant width but any classical nondeterministic OBDD needs non-constant width. The Section 5 contains our results on hierarchies on the sublinear (5.1) and larger (5.2) widths of deterministic and nondeterministic OBDDs.

The proofs of lower bounds results (Theorem \ref{th-DOBDD1} and Lemma \ref{peq1s}) are based on Pigeonhole principle. The lower bound on Theorem \ref{th-POBDD1} uses the technique of Markov chains. 

\section{Preliminaries}

We refer  to \cite{Weg00} for more information on branching programs. The main model investigated throughout the paper is OBDD (Ordered Binary Decision Diagram), a restricted version of branching program.

In this paper we use following notations for vectors. We use subscript for enumerating elements of vector and strings and superscript for enumerating vectors and strings.  For a binary string $\nu$, $\#_1(\nu)$ and $\#_0(\nu)$ are the number of $1$'s and $0$'s in $\nu$, respectively. We denote $\#^k_{0}(\nu)$ and $\#^k_{1}(\nu)$ to be the numbers of $1$'s and $0$'s in the first $k$ elements of string $\nu$, respectively.

For a given $ n>0 $, a probabilistic OBDD $ P_n $ with width $d$, defined on $ \{0,1\}^n $, is a 4-tuple
$
	P_n=(T,v^0,Accept{\bf,\pi}),
$
where
\begin{itemize}
	\item $ T = \{ T_j : 1 \leq j \leq n \mbox{ and } T_j = (A_j(0),A_j(1))  \} $ is an ordered pairs of (left) stochastic matrices representing the transitions, where, at the $j$-th step, $ A_j(0) $ or $ A_j(1) $, determined by the corresponding input bit, is applied.
	\item $ v^0 $ is a zero-one initial stochastic vector (initial state of $P_n$).
	\item $ Accept \subseteq \{1,\ldots,d\} $ is the set of accepting nodes.
\item $ \pi $ is a permutation of $ \{1,\ldots,n\} $ defining the order of testing the input bits.
\end{itemize}

  For any given input $ \nu \in  \{0,1\}^n $, the computation of $P_n$ on $\nu$ can be traced by a stochastic vector which is initially $ v^0 $. In each step $j$, $1 \leq j \leq n$, the input bit $ x_{\pi(j)} $ is tested and then the corresponding stochastic operator is applied:
\[
	v^j = A_{j}(x_{\pi(j)}) v^{j-1},
\]
where  $ v^j $ represent the probability distribution vector of nodes after the $ j$-th steps, $ 1 \leq j \leq n $.
The accepting probability of $ P_n $ on $ \nu $ is 
\[
	\sum_{i \in Accept} v^n_i.
\]

We say that a function $f$ is computed by $ P_n $ with bounded error if there exists an $ \varepsilon \in (0,\frac{1}{2}] $ such that $ P_n $ accepts all inputs from $ f^{-1}(1) $ with a probability at least $ \frac{1}{2}+\varepsilon $ and $ P_n $ accepts all inputs from $ f^{-1}(0) $ with a probability at most $ \frac{1}{2}-\varepsilon $. We say that $P_n$ computes $f$ \textit{exactly} if $ \varepsilon =1/2 $. 

A deterministic OBDD is a probabilistic OBDD restricted to use only 0-1 transition matrices. In the other words, the system is always being in a single node and, from each node, there is exactly one outgoing transition for each tested input bit.

A nondeterministic OBDD (NOBDD) can have the ability of making more than one outgoing transition for each tested input bit from each node and so the program can follow more than one computational path and if one of the path ends with an accepting node, then the input is accepted (rejected, otherwise). 

\begin{itemize}
\item 

An OBDD is called {\em stable} if each transition set $T_j$ is identical for each level.
\item

An OBDD is called ID (ID-OBDD) if the input bits are tested in order $\pi=(1,2,\dots, n)$.
If a {\em stable}  ID-OBDD has a fixed width and transition rules for each $ n $, then it can be considered as a realtime finite automaton.

\end{itemize}

Quantum computation is a generalization of classical computation \cite{Wat09A}. Therefore, each quantum model can simulate its probabilistic counterparts. In some cases, on the other hand, the quantum models are defined in a restricted way, e.g., using only unitary operators during the computation followed by a single measurement at the end, and so they may not simulate their probabilistic counterparts. Quantum automata literature has a lot of such results such as \cite{KW97,AF98,BC01B}. A similar result was also given for OBDDs in \cite{SS05}, in which a function with a small size of deterministic OBDD was given but the quantum OBDD defined in a restricted way needs exponential size to solve this function.

Quantum OBDDs that defined with the general quantum operators, i.e. superoperator \cite{Wat03,Wat09A,YS11A}, followed by a measurement on the computational basis at the end can simulate its classical counterpart with the same size and width. So we can always follow that any quantum class contains its classical counterpart. 

In this paper, we follow our quantum results based stable ID-OBDDs, which are realtime quantum finite automata. But, we give the definition of quantum OBDDs for the completeness of the paper.

A quantum OBDD is the same as a probabilistic OBDD with the following modifications:
\begin{itemize}
	\item The state set is represented by a $ d $-dimensional Hilbert space over field of complex numbers. The initial one is $ \ket{\psi}_0=\ket{q_0}$ where $ q_0 $ corresponds to the initial node.
	 
	 \item Instead of a stochastic matrix, we apply a unitary matrix in each step. That is, $ T = \{ T_j : 1 \leq j \leq n \mbox{ and } T_j = (U_j^0,U_j^1)  \} $, where, at the $j$-th step, $ U_j^0 $ or $ U_j^1 $, determined by the corresponding input bit, is applied,
	 
	 \item At the end, we make a measurement on the computational basis.
\end{itemize}

The state of the system is updated as follows after the $ j$-th step:
\[
	\ket{\psi}_j = U_j^{x_{\pi(j)}} (\ket{\psi}_{j-1}),
\]
where $ \ket{\psi}_{j-1} $ and $ \ket{\psi}_j $ represent the state of the system after the $ (j-1)$-th  and $ j$-th steps, respectively, where $ 1 \leq j \leq n $.

The accepting probability of the quantum program on $ \nu $ is calculated from $\ket{\psi}_n=(z_1, \dots, z_d)$ as
\[
	\sum_{i \in Accept} |z_i|^2.
\]

\section{Exact Quantum OBDDs.}

In \cite{AY12}, Ambainis and Yakary{\i}lmaz defined a new family of unary promise problems: For any $ k>0 $, $ A^k = (A^k_{yes},A^k_{no}) $ such that $ A^k_{yes} = \{ a^{(2i)2^k} : i \geq 0 \} $ and $ A^k_{no} = \{ a^{(2i+1)2^k} : i \geq 0 \} $. They showed that each member of this family ($A^k$) can be solved exactly by a 2-state realtime quantum finite automaton (QFA) but any exact probabilistic finite automaton (PFA) needs at least $ 2^{k+1} $ states. Recently, Rashid and Yakary{\i}lmaz \cite{RY14A} showed that bounded-error realtime PFAs also needs at least $ 2^{k+1} $ states for solving $A^k$.\footnote{The same result is also followed for two-way nondeterministic finite automata by Geffert and Yakary{\i}lmaz \cite{GY14A}.} Based on this promise problem, we define a partial function:
\[
	\mathtt{PartialMOD^k_n} (\nu)  = \left\lbrace \begin{array}{lcll}
		1 & , & \mbox{if } \#_1(\nu) = 0 & (mod\textrm{ }2^{k+1}) \\
		0 & , & \mbox{if } \#_1(\nu) = 2^{k} & (mod\textrm{ }2^{k+1}) \\
		* & , & \multicolumn{2}{l}{\mbox{otherwise}} \\	
\end{array}	 \right. ,
\]
where the function is not defined for the inputs mapping to ``*''. We call the inputs where the function takes the value of 1 (0) as yes-instances (no-instances).

\begin{theorem}
	For any $k \geq 0$, $ \mathtt{PartialMOD^k_n} $ can be solved by a stable quantum ID-OBDD with width 2 exactly.
\end{theorem}

The OBDD can be construct by the same way as QFA, which solves promise problem $ A^k $ \cite{AY12}.

We show that the width of deterministic, or bounded-error stable probabilistic OBDDs that solve $ \mathtt{PartialMOD^k_n} $  cannot be less than $ 2^{k+1} $.

\begin{remark} 
 Note that, a proof for deterministic OBDD is not similar to the proof for automata because potentially nonstability can give profit. Also this proof is different from proofs for total functions (for example, $ MOD_p$) due to the existence of incomparable inputs.
Note that, classical one-way communication complexity techniques also fail for partial functions  (for example, it can be shown that communication complexity of $ \mathtt{PartialMOD^k_n} $ is $1$), and we need to use more careful analysis in the proof. 
\end{remark}

 Deterministic stable ID-OBDD with width $2^{k+1}$ for $ \mathtt{PartialMOD^k_n} $ can be easy constructed. We left open the case for bounded-error non-stable probabilistic OBDDs.

\begin{theorem}
\label{th-DOBDD1}
	For any $k\geq 0$, there are infinitely many $n$ such that any deterministic  OBDD computing the partial function $ \mathtt{PartialMOD^k_n} $ has width at least $2^{k+1}$.
\end{theorem}

\begin{proof}
Let $\nu\in\{0,1\}^n, \nu=\sigma\gamma$. We call  $\gamma$ valid for $\sigma$ if $\nu \in (\mathtt{PartialMOD^k_n})^{-1}(0) \cup (\mathtt{PartialMOD^k_n})^{-1}(1).$ We call two substrings $\sigma'$ and $\sigma''$ comparable  if for all $\gamma$ it holds that $\gamma$ is valid for $\sigma'$  iff $\gamma$ is valid for $\sigma''$.
We call two substrings $\sigma'$ and $\sigma''$ nonequivalent if they are comparable and there exists a valid substring $\gamma$ such that   $\mathtt{PartialMOD^k_n}(\sigma'\gamma)\neq\mathtt{PartialMOD^k_n}(\sigma''\gamma)$ .  

Let $P$ be a deterministic OBDD computing the partial function $ \mathtt{PartialMOD^k_n} $. Note that paths associated with nonequivalent strings must lead to different nodes. Otherwise, if $\sigma$ and $\sigma'$ are nonequivalent, there exists a valid string $\gamma$ such that $\mathtt{PartialMOD^k_n}(\sigma\gamma)\neq\mathtt{PartialMOD^k_n}(\sigma'\gamma)$ and computations on these inputs lead to the same final node.

Let $N=2^k$ and $\Gamma=\{\gamma: \gamma\in \{0,1\}^{2N-1}, \gamma=0\dots  0 1\dots 1\}$.
We will naturally identify any string $\nu$  with the element $a=\#_1(\nu) \pmod{2N}$ of additive group $\mathbb{Z}_{2N}$. We call two strings of the same length different if the numbers of ones by modulo $2N$ in them are different. We denote by $\rho(\gamma^1,\gamma^2)=\gamma^1-\gamma^2$ the distance between numbers $\gamma^1, \gamma^2.$ Note that $\rho(\gamma^1,\gamma^2)\neq \rho(\gamma^2,\gamma^1) $.

Let the width of $P$ is $t<2N.$ On each step $i$ ($i=1,2,\dots$) of the proof we will count the number of different strings,  which lead to the same node 
(denote this node $v_i$). On $i$-th step we consider $(2N-1)i$-th level of $P$.

Let $i=1$.
By pigeonhole principle there exist two different strings $\sigma^1$ and $ \sigma^2$ from the set $\Gamma$ such that corresponding  paths lead to the same node $v_1$ of the $(2N-1)$-th level of $P$.  Note that $\rho(\sigma^1, \sigma^2)\neq N,$ because  in this case $\sigma^1$ and $\sigma^2$ are nonequivalent and cannot lead to the same node.

We will show  by induction that on each step of the proof the number of different strings which lead to the same node increases.

Step 2. By pigeonhole principle there exist two different strings $\gamma^1$ and $ \gamma^2$ from the set $\Gamma$ such that corresponding  paths from the node $v_1$ lead to the same node $v_2$ of the $(2N-1)2$-th level of $P$.  In this case, the strings $\sigma^1 \gamma^1,\sigma^2 \gamma^1,\sigma^1 \gamma^2 $, and $\sigma^2 \gamma^2$ lead to the node $v_2$. Note that $\rho(\gamma^1, \gamma^2)\neq N,$ because  in this case $\sigma^1\gamma^1$ and $\sigma^1\gamma^2$ are nonequivalent and cannot lead to the same node. 

Adding  the same number does not change the distance between the numbers, so we have
   \[\rho(\sigma^1+\gamma^1, \sigma^2+\gamma^1)=\rho(\sigma^1,\sigma^2)\]
   and
\[\rho(\sigma^1+\gamma^2, \sigma^2+\gamma^2)=\rho(\sigma^1,\sigma^2).\]
Let  $\gamma^2>\gamma^1$. Denote $\Delta=\gamma^2-\gamma^1 $. Let us count the number of different numbers among $\sigma^1+\gamma^1,$ $\sigma^2+\gamma^1  $, $\sigma^1+\gamma^1+\Delta  $, and $\sigma^2+\gamma^1+\Delta  $.
Because $\sigma^1$ and $\sigma^2$ are different and $\rho(\sigma^1, \sigma^2)\neq N$,  the numbers from the pair $\sigma^1+\gamma^1  $, and $\sigma^2+\gamma^1  $  coincide with corresponding numbers from the pair $\sigma^1+\gamma^1+\Delta  $ and $\sigma^2+\gamma^1+\Delta  $ iff $\Delta=0 \pmod {2N}$. But  $\Delta \neq 0 \pmod {2N}$ since the  numbers $\gamma^1 $ and $ \gamma^2$ are different and  $\gamma^1, \gamma^2<2N.$ The numbers  $\sigma^1+\gamma^1+\Delta  $ and $\sigma^2+\gamma^1+\Delta   $  cannot be a permutation of numbers   $\sigma^1+\gamma^1  $ and $\sigma^2+\gamma^1  $  since $\rho(\gamma^1, \gamma^2)\neq N$ and    $\rho(\sigma^1, \sigma^2)\neq N$.
In this case, at least 3 numbers from $\sigma^1+\gamma^1  $, $\sigma^2+\gamma^1  ,$ $\sigma^1+\gamma^2  $, and $\sigma^2+\gamma^2  $ are different.

Step of induction.  Let the numbers $\sigma^1, \dots, \sigma^i$ be different on the step $i-1$ and the corresponding paths lead to the same node $v_{i-1}$ of the $(2N-1)(i-1)$-th level of $P$. 

By pigeonhole principle there exist two different strings $\gamma^1 $ and $ \gamma^2$ from the set $\Gamma$ such that  the corresponding  paths from the node $v_{i-1}$ lead to the same node $v_i$ of the $(2N-1)i$-th level of $P$.  So paths $\sigma^1\gamma^1,\dots, \sigma^i\gamma^1,\sigma^1\gamma^2,\dots, \sigma^i\gamma^2 $ lead to the same node $v_i$. Let us estimate a number of different strings among them. Note that  $\rho(\gamma^1, \gamma^2)\neq N$ since, in this case, the strings $\sigma^1\gamma^1$ and $ \sigma^1\gamma^2$ are nonequivalent but lead to the same node. 

The numbers $\sigma^1, \dots,\sigma^{i}$ are different and $\rho(\sigma^l, \sigma^j)\neq N$ for each pair $(l,j) $ such that $ l\neq j$.
Let $\sigma^1 < \dots <\sigma^i$. We will show that among $\sigma^1+\gamma^1  ,\dots, \sigma^i+\gamma^1   $ and $\sigma^1+\gamma^1+\Delta  , \dots, \sigma^i+\gamma^1+\Delta  $ at least $i+1$ numbers are different.

The sequence of numbers $\sigma^1+\gamma^1  ,\dots, \sigma^i+\gamma^1  $ coincide with the the sequence $\sigma^1+\gamma^1+\Delta  , \dots, \sigma^i+\gamma^1+\Delta  $ iff $\Delta=0 \pmod {2N}$. But $\Delta\neq 0 \pmod {2N}$ since $\gamma^1$ and $ \gamma^2$ are different and  $\gamma^1, \gamma^2<2N.$
 
Suppose that the sequence $\sigma^1+\gamma^1+\Delta  , \dots, \sigma^i+\gamma^1+\Delta  $  is a permutation of the sequence  $\sigma^1+\gamma^1  $,$\dots,$ $\sigma^i+\gamma^1  $. In this case, 
we have numbers $a_0, \dots, a_r$ from $\mathbb{Z}_{2N}$  such that all $a_j$ are from the sequence $\sigma^1+\gamma^1  $,$\dots,$ $\sigma^i+\gamma^1  $, $a_0=a_r=\sigma^1+\gamma^1  $, and $a_j=a_{j-1}+\Delta  $, where $j=1, \dots, r$. 
In this case, $r\Delta=2Nm.$
Because $N=2^k$, $\Delta<2N$, and $\Delta\neq N$  we have that $r$ is even. For $z=r/2$ we have $z\Delta=Nm$.
Since all numbers from $\sigma^1+\gamma^1,\dots, \sigma^i+\gamma^1$ are different, we have that $\rho(a_0,a_z)=N$.
So we have that $a_0 $ and $ a_z$ are  nonequivalent but the corresponding strings lead to the same node $v_i$. 
So after $i$-th step, we have that at least $i+1$ different strings lead to the same node $v_i.$

On the $N$-th step, we have that $N+1$ different strings lead to the same node $v_N$. Among these strings, there must be at least two nonequivalent strings. Thus we can follow that $P$ cannot compute the function $\mathtt{PartialMOD^k_n}$ correctly.  
\qed\end{proof}

\begin{theorem}
\label{th-NOBDD1}
	For any $k\geq 0$, there are infinitly many $n$ such that any nondeterministic  OBDD computing the partial function $ \mathtt{PartialMOD^k_n} $ has width at least $2^{k+1}$.
\end{theorem}

The proof is similar to deterministic case with the following modifications. 
Let $P$ -- NOBDD, that computes $ \mathtt{PartialMOD^k_n} $. 
We consider only accepting pathes in $ P $. Note that if $P$ computes $ \mathtt{PartialMOD^k_n} $ correctly then accepting paths associated with nonequivalent strings can not pass through the same nodes. Denote $N=2^k$.  Let $\Gamma=\{\gamma: \gamma\in \{0,1\}^{2N-1}, \gamma=\underbrace{0\dots  0}_{2N-1-j} \underbrace{1\dots 1}_{j}, j=0, \dots,2N-1\}$.

 Denote $V_{l}$ -- a set of nodes on the $l$-th level of $ P $. Assume that the width of $P$ is less than $2N$, that is, $|V_l|<2N$ for each $ l=0, \dots,n $.

Denote $V_l(\gamma, V)$ -- a set of nodes of $l$-th level through which  accepting paths, leading from the nodes of set $ V $ and corresponding to string $\gamma$, pass.

On step $i$ ($ i=1,2,\dots $) of the proof we consider $(2N-1)i$-th level of $P$.
Because $|V_{2N-1}|<2N$, on the first step of the proof  we have that there exists two different strings $\sigma^1,\sigma^2$ from the set $\Gamma$ such that $ V_{2N-1}(\sigma^1, V_0) \cap V_{2N-1}(\sigma^2, V_0) \neq \emptyset$.  Denote this nonempty intersection $G_1$. Next we continue our proof considering only  accepting paths which  pass through the nodes of $G_1.$

Using the same ideas as for deterministic case we can show that the number of strings with different numbers of ones by modulo $2N$, such that corresponding accepting paths  lead to the same set of nodes, increases with each step of the proof. 

On the $N$-th step of the proof we will have that there must be two different nonequivalent strings such that corresponding accepting paths  lead to the same set  $G_N$ of nodes.  That implies that $P$ does not compute the function $\mathtt{PartialMOD^k_n}$ correctly. 




\begin{theorem}
\label{th-POBDD1}
For any $k\geq 0$, there are infinitely many $n$ such that any stable probabilistic OBDD computing $ \mathtt{PartialMOD^k_n} $ with bounded error has width at least $2^{k+1}$. 
\end{theorem}

The proof of the Theorem is based on the technique of Markov chains and the details are given in Appendix \ref{app-POBDD1}.

\section{Nondeterministic Quantum and Classical OBDDs.}

In \cite{YS10A}, Yakary{\i}lmaz and Say showed that nondeterministic QFAs can define a super set of regular languages, called exclusive stochastic languages \cite{Paz71}. This class contains the \textit{complements} of some interesting languages:
$ PAL = \{ w \in \{0,1\}^* : w = w^r \} $, where $w^r$ is a reverse of $w$,   $ O = \{ w \in \{0,1\}^* : \#_1(w)=\#_0(w)\} $, $ SQUARE = \{ w \in \{0,1\}^* : \#_1(w)=(\#_0(w))^2\} $, and $ POWER = \{ w \in \{0,1\}^* : \#_1(w)=2^{\#_0(w)}\} $.

Based on these languages, we define three symmetric functions for any input $ \nu \in \{0,1\}^n $:
	 \[ \mathtt{NotO_n}(\nu) = \left\lbrace \begin{array}{lll}
		0 & , & \mbox{if } \#_0(\nu) = \#_1(\nu) \\
		1 & , & \mbox{otherwise }
\end{array}	 \right., \]
	\[ \mathtt{NotSQUARE_n}(\nu) = ~~~~ \left\lbrace \begin{array}{lll}
		0 & , & \mbox{if } (\#_0(\nu))^2 = \#_1(\nu) \\
		1 & , & \mbox{otherwise }
\end{array}	 \right.,  \]
	\[  \mathtt{NotPOWER_n}(\nu) = \left\lbrace \begin{array}{lll}
		0 & , & \mbox{if } 2^{\#_0(\nu)} = \#_1(\nu) \\
		1 & , & \mbox{otherwise }
\end{array}	 \right..  \]

\begin{theorem}
	Boolean Functions $ \mathtt{NotO_n} $, $ \mathtt{NotSQUARE_n} $, and $ \mathtt{NotPOWER_n} $ can be solved by a nondeterministic quantum OBDD with constant width.
\end{theorem}

For all these three functions, we can define nondeterministic quantum (stable ID-) OBDD with constant width based on nondeterministic QFAs for languages $ O $, $SQUARE$, and $POWER$, respectively \cite{YS10A}.

The complements of $ PAL, O, SQUARE $ and $POWER$ cannot be recognized by classical nondeterministic finite automata. But, for example, the function version of the  complement of $PAL$, $ \mathtt{NotPAL_n} $, which returns 1 only for the non-palindrome inputs, is quite easy since it can be solved by a deterministic OBDD with width $3$. Note that the order of such an OBDD is not natural $(1,\dots,n)$. However, as will be shown soon, this is not the case for the function versions of the complements of the other three languages.

\begin{theorem}\label{nondetermenisticBoundTheorem}
There are infinitely many $n$ such that any NOBDD $P_n$ computing $ \mathtt{NotO_n} $ has width at least $\lfloor\log n\rfloor -1 $.
\end{theorem}

The proof of Theorem is based on the complexity properties of Boolean function $\mathtt{NotO_n}$. At first we will discuss complexity properties of this function in Lemma \ref{leqN}. After that we will prove claim of Theorem.
	
\begin{lemma}\label{leqN}
There are infinitely many $n$ such that any OBDD computing $ \mathtt{NotO_n} $ has width at least $n/2 + 1$. (For the proof see \ref{app-leqN}).
\end{lemma}
{\em Proof of Theorem \ref{nondetermenisticBoundTheorem}}\quad\quad
 Let function $\mathtt{NotO_n}$ is computed  by  $NOBDD$ $P_n$ of  width $d$, then by the same way as in the proof of Theorem \ref{th-NOBDD1} we have
$d\geq\log (n/2 + 1)>\log n -1.$
\Endproof

By the same way we can show that there are infinitely many $n$ such that any NOBDD $P_n$ computing function $\mathtt{NotSQUARE_n}$ has width at least  $ \Omega(\log (n)) $ and any NOBDD $P'_n$ computing function $\mathtt{NotPOWER_n}$ has width at least  $ \Omega(\log\log (n)) $.

\section{Hierarchies for Deterministic and Nondeterministic OBDDs}

We denote $\mathsf{OBDD^d}$ and $\mathsf{NOBDD^d}$ to be the sets of Boolean functions that can be computed by $OBDD$ and $NOBDD$ of width $d=d(n)$, respectively, where $n$ is the number of variables.
In this section, we present some width hierarchies for $\mathsf{OBDD^d}$ and $\mathsf{NOBDD^d}$. Moreover, we discuss relations between these classes.
We consider $\mathsf{OBDD^d}$ and $\mathsf{NOBDD^d}$ with small (sublinear) widths and large widths.

\subsection{Hierarchies and relations for small width OBDDs}

We have the following width hierarchy for deterministic and nondeterministic models.

\begin{theorem}\label{smallwh}
For any integer $n$, $d=d(n)$, and $1<d\leq n/2$, we have
\begin{eqnarray}
\mathsf{OBDD^{d-1}} & \subsetneq  & \mathsf{OBDD^{d}}\label{st1s} \mbox{ and}
\\
\mathsf{NOBDD^{d-1}} & \subsetneq &  \mathsf{NOBDD^{d}}\label{st2s} .
\end{eqnarray}
\end{theorem}

\paragraph{Proof of Theorem \ref{smallwh}.}

It is obvious that $\mathsf{OBDD^{d-1}}\subseteq \mathsf{OBDD^{d}}$ and $\mathsf{NOBDD^{d-1}}\subseteq \mathsf{NOBDD^{d}}$. Let us show the inequalities of these classes. For this purpose we use the complexity properties of known Boolean function $\mathtt{MOD^k_{n}}$. 

Let $k$ be a number such that $1<k\leq n/2$. For any given input $\nu\in\{0,1\}^n$,

\begin{displaymath}
\mathtt{MOD^k_{n}}(\nu) = \left\{ \begin{array}{ll}
1, & \textrm{if } \#_{1}(\nu)= 0\textrm{ } (mod\textrm{ } k), \\
0, & \textrm{otherwise}
\end{array} \right. .
\end{displaymath}

\begin{lemma}\label{peq2s}

There is an OBDD (and so a NOBDD) $P_n$ of width $d$ which computes Boolean function $\mathtt{MOD^k_{n}}$ and $d=k$.
\end{lemma}

{\em Proof}\quad\quad 
In each level, $P_n$ counts  number of $1$'s by modulo $k$. $P_n$ answers $1$ iff  the number in the last step is zero. It is clear that the width of $P_n$ is $k$. 
\Endproof

\begin{lemma}\label{peq1s}
Any OBDD and NOBDD computing $\mathtt{MOD^k_{n}}$ has width at least $ k$.
\end{lemma}

{\em Proof}\quad\quad
The proof is based on Pigeonhole principle. Let $P$ be a deterministic OBDD computing the function $ \mathtt{MOD^k_{n}} $. 
For each input $\nu$ from $ (\mathtt{MOD^k_{n}})^{-1}(1) $ there must be exactly one path in $P$ leading from source node to accepting node.  Let us consider $k$ inputs $\{\nu^1, \nu^2,\dots, \nu^{k} \}$ from this set such that the last $k$ bits in $\nu^j (j=1,\dots, k)$ contains exactly $j$ 1's and $(k-j)$ 0's. Let us consider the $(n-k)$-th level of $P$.
The acceptance paths for different inputs from $\{\nu^1, \nu^2,\dots, \nu^{k} \}$ must pass trough  different nodes  of 
 the $(n-k)$-th level of $P$. So the width of the $(n-k)$-th level of $P$ is at least $k$.

The proof for nondeterministic case is similar to deterministic one. For each input from  $ (\mathtt{MOD^k_{n}})^{-1}(1) $  for the function $ \mathtt{MOD^k_{n}} $ there must be  at least one path in $P$ leading from the source node to an accepting node labelling this input.  The accepting paths for different inputs from the set $\{\nu^1,\nu^2,\dots, \nu^{k} \}$  must go through different nodes of $(n-k)$-th level of $P$.
\Endproof

Boolean function $\mathtt{MOD^d_{n}} \in \mathsf{OBDD^{d}}$ and  $\mathtt{MOD^d_{n}} \in \mathsf{NOBDD^{d}}$ due to Lemma \ref{peq2s} and Boolean function $\mathtt{MOD^d_{n}} \not\in \mathsf{OBDD^{d-1}}$ and $\mathtt{MOD^d_{n}} \not\in \mathsf{NOBDD^{d-1}}$ due to  Lemma \ref{peq1s}.
This  completes the proof of the Theorem \ref{smallwh}. \qed

We have the following relationships between deterministic and nondeterministic models.

\begin{theorem}\label{smallwh2}
For any integer $n$, $d=d(n)$, and $d'=d'(n)$ such that $d\leq n/2$ and $O(\log^2 d \log\log d)<d'\leq d-1$, we have
\begin{eqnarray}
\mathsf{NOBDD^{\lfloor\log (d)\rfloor}} \subsetneq  \mathsf{OBDD^{d}} \mbox{ and} \label{st4s}
\\
 \mathsf{OBDD^{d}} \textrm{ and }   \mathsf{NOBDD^{d'}} \textrm{ are not comparable.} \label{st3s}
\end{eqnarray}

\end{theorem}

\paragraph{Proof of Theorem \ref{smallwh2}}
We start with (\ref{st4s}). 
By the same way as in the proof of Theorem 
\ref{th-NOBDD1}, we can show that $\mathsf{NOBDD^{\lfloor\log (d)\rfloor}}\subseteq\mathsf{OBDD^{d}}$ and, from Lemma \ref{peq1s},  we know that $\mathtt{MOD^d_{n}} \notin \mathsf{NOBDD^{\lfloor\log (d)\rfloor}}$. Then we have $\mathsf{OBDD^{d}}\neq \mathsf{NOBDD^{\lfloor\log (d)\rfloor}}$.

We continue with (\ref{st3s}). Let $k$ be even and $1<k\leq n$. For any given input $\nu\in\{0,1\}^n$,
\begin{displaymath}
\mathtt{NotO^k_{n}}(\nu) = \left\{ \begin{array}{ll}
0, & \textrm{if } \#^k_{0}(\nu)= \#^k_{1}(\nu)=k/2 \\
1, & \textrm{otherwise}
\end{array} \right. .
\end{displaymath}
Note that function $\mathtt{NotO^n_{n}}$ is identical to $\mathtt{NotO_{n}}$.
\begin{lemma}\label{pneq1s}
Any OBDD computing $\mathtt{NotO^k_{n}}$ has width at least $ k/2+1$.
\end{lemma}
{\em Proof}\quad\quad 
The proof can be followed by the same technique given in the proof of Lemma \ref{leqN}.
\Endproof
\begin{lemma}\label{pneq4s}
There is NOBDD $P_n$ of width $d$ that computes Boolean function $\mathtt{NotO^k_{n}}$ and $d\leq O(\log^2 k \log\log k)$.  (For the proof see \ref{app-pneq4s}).
\end{lemma}

Remember that $O(\log^2 d \log\log d)\leq d'\leq d-1$, and, by Lemma \ref{peq2s} and Lemma \ref{peq1s}, we have $\mathtt{MOD^d_{n}} \in \mathsf{OBDD^{d}}$ and $\mathtt{MOD^d_{n}} \notin\mathsf{NOBDD^{d'}}$; by Lemma \ref{pneq4s}, we have $\mathtt{NotO_{n}^{2d-1}}\in\mathsf{NOBDD^{d'}}$; and, by Lemma \ref{pneq1s}, we have $\mathtt{NotO_{n}^{2d-1}} \notin \mathsf{OBDD^{d}}$. Therefore, we cannot compare these classes and so we can follow Theorem \ref{smallwh2}. \qed

\subsection{Hierarchies and relations for large width OBDDs}

In this section, we consider OBDDs of large width. We obtain some hierarchies which are different from the ones in the previous section (Theorem \ref{smallwh}).

\begin{theorem}\label{hi1}
For any integer $n$, $d=d(n)$, $16\leq d \leq 2^{n/4}$, we have
\begin{eqnarray}
\mathsf{OBDD^{\lfloor d/8 \rfloor-1}}  \subsetneq  \mathsf{OBDD^{d}} \mbox{ and }\label{st1}\\
\mathsf{NOBDD^{\lfloor d/8 \rfloor-1}}  \subsetneq  \mathsf{NOBDD^{d}},\label{st2}
\end{eqnarray}

\end{theorem}

\paragraph{Proof of Theorem \ref{hi1}.}

It is obvious that $\mathsf{OBDD^{\lfloor d/8 \rfloor-1}}\subseteq \mathsf{OBDD^{d}}$ and $\mathsf{NOBDD^{\lfloor d/8 \rfloor-1}} $ $ \subseteq  \mathsf{NOBDD^{d}}$. 

 We define Boolean function $\mathtt{EQS^k_{n}}$ as a modification of  Boolean function  {\em Shuffled Equality} which was defined in \cite{ak96} and \cite{A97}. The proofs of inequalities are based on the complexity properties of $\mathtt{EQS^k_{n}}$.

Let $k$ be multiple of $4$ such that $4\leq k\leq 2^{n/4}$. The Boolean Function $\mathtt{EQS_{n}}$ depends only on the first $k$ bits. 

For any given input $\nu\in\{0,1\}^n$, we can define two binary strings $\alpha(\nu)$ and $\beta(\nu)$ in the following way. We call odd bits of the input {\em marker bits} and even bits  {\em value bits}. For any $i$ satisfying $1\leq i \leq k/2$, the value bit $\nu_{2i}$ belongs to $\alpha(\nu)$  if the corresponding marker bit $\nu_{2i-1}=0$ and $\nu_{2i}$ belongs to $\beta(\nu)$ otherwise. 

\begin{displaymath}
\mathtt{EQS^k_{n}}(\nu) = \left\{ \begin{array}{ll}
1, & \textrm{if } \alpha(\nu)= \beta(\nu)\\
0, & \textrm{otherwise}
\end{array} \right. .
\end{displaymath}

\begin{lemma}\label{peq1}

There is OBDD $P_n$ of width $8\cdot 2^{k/4}-5$ which computes Boolean function $\mathtt{EQS^k_{n}}$. (For the proof see \ref{app-peq1}).
\end{lemma}

\begin{lemma}\label{peq2}
There are infinitely many $n$ such that any OBDD and NOBDD $P_n$ computing $ \mathtt{EQS^k_{n}} $ has width at least $2^{k/4}$. (For the proof see \ref{app-peq2}). (For the proof see \ref{app-peq2})
\end{lemma}

Boolean function $\mathtt{EQS^{4\lceil\log (d+5)\rceil-12}_n} \in \mathsf{OBDD^{d}}$ and  $\mathtt{EQS^{4\lceil\log (d+5)\rceil-12}_n} \in \mathsf{NOBDD^{d}}$ due to Lemma \ref{peq1}.

Boolean function $\mathtt{EQS^{4\lceil\log (d+5)\rceil-12}_n} \not\in \mathsf{OBDD^{\lfloor d/8 \rfloor-1}}$ and $\mathtt{EQS^{4\lceil\log (d+5)\rceil-12}_n} \not\in \mathsf{NOBDD^{\lfloor d/8 \rfloor-1}}$ due to Lemma \ref{peq2}. So $\mathsf{OBDD^{\lfloor d/8 \rfloor-1}}\neq \mathsf{OBDD^{d}}$ and $\mathsf{NOBDD^{\lfloor d/8 \rfloor-1}} $  $ \neq \mathsf{NOBDD^{d}}$. These inequalities prove Statements (\ref{st1}) and (\ref{st2}) and complete the proof of the Theorem \ref{hi1}. \qed
   
In the following theorem, we present a relationship between deterministic and nondeterministic models.

\begin{theorem}\label{hi3}
For any integer $n$, $d=d(n)$, and $d'=d'(n)$ satisfying $ d\leq 2^{n/4}$ and $O(\log^4 (d+1) \log\log (d+1))<d'<d/8-1$,  we have
\begin{eqnarray}
\mathsf{NOBDD^{\lfloor\log (d)\rfloor}} \subsetneq  \mathsf{OBDD^{d}}\mbox{ and } \label{st4}\\
\mathsf{OBDD^{d}} \textrm{ and }   \mathsf{NOBDD^{d'}} \textrm{ are not comparable.}\label{st3}
\end{eqnarray}

\end{theorem}

\paragraph{Proof of Theorem \ref{hi3}.}
We start with (\ref{st4}). By the same way as in proof of Theorem 
\ref{th-NOBDD1}, we can show that $\mathsf{NOBDD^{\lfloor\log (d)\rfloor}}\subseteq\mathsf{OBDD^{d}} $. By Lemma \ref{peq2},  we have $\mathtt{EQS_n^{4\lceil\log (d+5)\rceil-12}} \notin \mathsf{NOBDD^{\lfloor\log (d)\rfloor}}$ which means $\mathsf{OBDD^{d}} $ $ \neq \mathsf{NOBDD^{\lfloor\log (d)\rfloor}}$.

Now we continue with (\ref{st3}).
We use the complexity properties of Boolean function $\mathtt{NotEQS^k_{n}}$, which is the inversion of $\mathtt{EQS^k_{n}}$. 



\begin{lemma}\label{pneq2}
There are infinitely many $n$ such that any OBDD $P_n$ computing $ \mathtt{NotEQS^k_{n}} $ has width at least $2^{k/4}$.
\end{lemma}

{\em Proof}\quad\quad
We can prove it by the same way as in the proof of Lemma \ref{peq2}.
\Endproof

\begin{lemma}\label{pneq1}
There is a NOBDD $P_n$ of width $d$ computing Boolean function $\mathtt{NotEQS^k_{n}}$ where $d\leq O(k^4 \log k)$. (For the proof see \ref{app-pneq1}).
\end{lemma}

Remember that $O(\log^4 (d+1) \log\log (d+1))\leq d'\leq \lfloor d/8\rfloor-1$, and, by Lemma \ref{peq2} and Lemma \ref{peq1}, we have $\mathtt{EQS_n^{4\lceil\log (d+5)\rceil-12}} \in \mathsf{OBDD^{d}}$ and $\mathtt{EQS_n^{4\lceil\log (d+5)\rceil-12}} \notin\mathsf{NOBDD^{d'}}$; by Lemma \ref{pneq1}, we have $\mathtt{NotEQS_n^{4\lceil\log(d)\rceil+4}}\in\mathsf{NOBDD^{d'}}$; and, by Lemma \ref{pneq2}, we have $\mathtt{NotEQS_n^{4\lceil\log(d)\rceil+4}} \notin \mathsf{OBDD^{d}}$. Therefore we cannot compare this classes and so we can follow Theorem \ref{hi3}. \qed
%


%


%
%



\newpage

\appendix

\section{The Proof of Theorem \ref{th-POBDD1}}
\label{app-POBDD1}

Assume that there is a stable probabilistic
OBDD $P_n$ of width $d < 2^{k+1}$ computing $ \mathtt{PartialMOD^k_n} $ with probability
$1/2+\varepsilon$ for a fixed $\varepsilon\in (0,1/2]$.  Let
$v^j=(v^j_1,\dots,v^j_d)$ be a probability distribution of
nodes of $P_n$ at the $j$-th level, where $v^j_i$ is the probability
of being in the $i$-th node at the $j$-th level. We can describe the
computation of $P_n$ on the input $\nu=\nu_1,\dots,\nu_n$ as follows:

\begin{itemize}

\item The computation of $P_n$ starts from the initial probability
distributions vector $v^0$.

\item On the $j$-th step, $1\leq j \leq n$, $P_n$ reads input $\nu_{i_j}$
and transforms the vector $v^{j-1}$ to $v^j=A v^{j-1}$, where
$A$ is the $q
\times q$ stochastic matrix, $A=A(0)$ if $\nu_{i_j}=0$ and $A=A(1)$ if
$\nu_{i_j}=1$.

\item After the last ($n$-th) step of the computation $P_n$ accepts the input
$\nu$ with probability $P_{acc}( \nu)=\sum_{j \in
Accept}v^n_j$. If $ \mathtt{PartialMOD^k_n} ( \nu)=1$ then we have $P_{acc}(
\nu)\geq 1/2 + \varepsilon$ and if $ \mathtt{PartialMOD^k_n} ( \nu)=0 $ then we have  $P_{acc}(
\nu)\leq 1/2 - \varepsilon$.

\end{itemize}

Without loss of generality, we assume that $P_n$ reads the inputs in the
natural order $x_1,\dots,x_n$.
We consider only inputs $\tilde\nu_n,\dots,\tilde\nu_1$ such that
$\tilde\nu_i=\tilde\nu_i^0 \tilde\nu_i^1,$ where
$\tilde\nu_i^0=\underbrace{0 \dots 0}_{n-i},$
$\tilde\nu_i^1=\underbrace{1 \dots 1}_{i}.$

For $i \in \{1,\dots, n\}$, we denote by $\alpha^i$ the probability
distribution after reading the part $ \tilde\nu_i^0$, i.e.
$\alpha^i=A^{n-i}(0)v^0 $.  There are only ones
in the $ \tilde\nu_i^1$, hence the computation after reading $
\tilde\nu_i^0$ can be described by a Markov chain. That is, $\alpha^i$
is the initial probability distribution for a Markov process and $A(1)$
is the transition probability matrix.

According to the classification of Markov chains described in the Section 2 of the book by Kemeny and Snell\footnote{J.~G. Kemeny and J.~L. Snell.
\newblock {\em Finite Markov Chains}.
\newblock Van Nostrand Company, INC, 1960.
}, the
states of the Markov chain are divided into ergodic and transient
states. An {\em ergodic set of states} is a set which a process cannot
leave once it has entered. A {\em transient set of states} is a set
which a process can leave, but cannot return once it has left. An {\em
ergodic state} is an element of an ergodic set. A {\em transient
state} is an element of a transient set.

An arbitrary Markov chain $C$ has at least one ergodic set. $C$ can be
a Markov chain without any transient set. If a Markov chain $C$ has
more than one ergodic set, then there is absolutely no interaction
between these sets. Hence we have two or more unrelated Markov chains
lumped together.  These chains can be studied separately. If a Markov
chain consists of a single ergodic set, then the chain is called an
{\em ergodic chain}. According to the classification mentioned above,
every ergodic chain is either regular or cyclic.

If an ergodic chain is regular, then for sufficiently high powers of
the state transition matrix, $A$ has only positive elements. Thus, no
matter where the process starts, after a sufficiently large number of
steps it can be in any state. Moreover, there is a limiting vector of
probabilities of being in the states of the chain, that does not
depend on the initial state.

If a Markov chain is cyclic, then the chain has a period $t$ and all
its states are subdivided into $t$ cyclic subsets $(t>1)$. For a given
starting state a process moves through the cyclic subsets in a
definite order, returning to the subset with the starting state after
every $t$ steps. It is known that after sufficient time has elapsed,
the process can be in any state of the cyclic subset appropriate for
the moment. Hence, for each of  $t$ cyclic subsets the $t$-th power
of the state transition matrix $A^t$ describes a regular Markov chain.
Moreover, if an ergodic chain is a cyclic chain with the period $t$, it
has at least $t$ states.

Let $C_1, \dots, C_l$ be cyclic subsets of states of Markov chain with periods $t_1, \dots, t_l$, respectively, and $D$ be the least common multiple of $t_1, \dots, t_l$.

\begin{lemma}
$D$ must be a multiple of $2^{k+1}$.
\end{lemma}
\begin{proof} 
Assume that $D$ is not a multiple of $2^{k+1}.$ After every $D$ steps, the process can be in any set of states containing the
 accepting state and the $D$-th power of $M$ describes a regular
 Markov chain. From the theory of Markov chains, we have
 that there exists an $\alpha_{acc}$ such that $\lim_{r \to
 \infty}\alpha^{r \cdot D}_{acc}=\alpha_{acc}$, where $\alpha_{acc}^i$ represents the probability of process being in accepting state(s) after $i$th step. Hence, for any $\varepsilon>0$,
 it holds that
 \[|\alpha^{r \cdot D}_{acc}-\alpha^{r' \cdot  D}_{acc}|<2\varepsilon\] for some big enough $r, r'$.
Since $D$ is not a multiple of $2^{k+1}$,  it can be represented as $D=m \cdot 2^l$ ($l\leq k$, $m$ is odd). For any odd $s$,  the number $s \cdot D$  is not a multiple of $2^{k+1}$.
$P_n$ is supposed to solve $ \mathtt{PartialMOD^k_n} $ with probability  $1/2+\varepsilon$ so we have
 $\alpha^{s\cdot m\cdot 2^l 2^{k-l+1}}_{acc}\geq 1/2+\varepsilon$ and
 $\alpha^{s\cdot m\cdot 2^l 2^{k-l}}_{acc}\leq 1/2-\varepsilon.$ This contradicts with the inequality above for big enough $s$. 
\qed\end{proof}

\begin{lemma}
	There is a circle of period $t$, where $t$ is a multiple of $2^{k+1}.$
\end{lemma}
\begin{proof}
The proof is followed from the facts that $D$ is a multiple of $2^{k+1}$ that implies existence of $t \in \{t_1, \dots, t_l \}$ such that $t$ is a power of 2. Among such $t$ there must be a multiple of $2^{k+1}$. Otherwise $D$ (the least common multiple of $t_1,\dots,t_l$) can not be a multiply of $2^{k+1}$. 
\qed\end{proof}

Since there is a circle of period $t$ where $t$ is  a multiple of $2^{k+1}$, we have $q\geq 2^{k+1}.$


\section{The Proof of Lemma \ref{leqN}}\label{app-leqN}

Let $n$ be an even integer and  $P_n$ be an OBDD that computes $ \mathtt{NotO_n} $. Assume that  $P_n$ uses the natural order. The proof is similar for any other order. Now we compute the width of level $n/2$.

Let $\Sigma=\{\sigma^i\in\{0,1\}^{n/2}:\sigma^{i}=\underbrace{1 \dots 1}_{i}\underbrace{0 \dots 0}_{n/2-i}, 0\leq i \leq n/2\}$ and $\sigma^{i} $ and $ \sigma^{j}$ be any pair such that $i\neq j$.
Assume that $P_n$ reaches node $v$ by input $\sigma^{i}$ and node $v'$ by input $\sigma^{j}$. 

We show that $v\neq v'$.  We define $\gamma\in \{0,1\}^{n/2}$ such that $\#_0(\gamma)=i$, which also means $\#_1(\gamma)=n/2-i$. 

Let us consider the computation on input $(\sigma^{i}, \gamma)$. Note that $P_n$ reaches $0$ from node $v$  using input $\gamma$ since
\[\#_1(\sigma^{i})+\#_1(\gamma)=n/2=\#_0(\sigma^{i})+\#_0(\gamma).\]

Let us consider the computation on input $(\sigma^{j}, \gamma)$. OBDD  $P_n$ reaches $1$ form node $v'$  using input $\gamma$ since

\[\#_1(\sigma^{j})+\#_1(\gamma)=j + (n/2-i)=n/2+ (j-i)\neq n/2.\]

It means $v\neq v'$. Note that $|\Sigma|=n/2+1$. Hence the width of level $n/2$ is at least $n/2 + 1$, which means that the width of $P_n$ is at least $n/2 + 1$.






\section{The Proof of Lemma \ref{pneq4s}}\label{app-pneq4s}
We construct the NOBDD $P_n$. We use the fingerprinting method given in \cite{Fre79}. 
Let $p_1,\dots,p_r$ be the first $r$ prime numbers satisfying that
\[
p_1 p_2 \cdots p_r>k,
\] where $r$ is the minimal value. Note that $r\leq \log k$ since $p_r\geq 2$.

The NOBDD $P_n$ consists of $r$ parallel parts, each of them corresponds to one of $p$ from $\{p_1,\dots,p_r\}$. In the first step, $P_n$ nondeterministically picks a $p$. Then, this branch counts  the number of $1$'s by modulo $p$. If it is not equal to $(k/2 \mod p)$ at the end, $P_n$ gives the answer of $1$ and $0$ otherwise. We need $p$ nodes for each value from $0$ to $p-1$.

By the Chinese Remainder Theorem\footnote{K.~Ireland and M.~Rosen.
\newblock {A Classical Introduction to Modern Number Theory, 2nd ed.}
\newblock {\em Springer-Verlag}, pages 34--38, 1990.
}, if $\#^k_{1}(\nu)\neq k/2$ then there are at least  one $p_i$ such that $\#^k_{1}(\nu)\neq k/2 \mod p_i$. This means at least one branch gives the answer of $1$.
But if $\#^k_{1}(\nu)= k/2$, then for all $p_i$ we have $\#^k_{1}(\nu)= k/2 \mod p_i$. 
This means each branch gives the answer of $0$.
Hence, $P_n$ computes $\mathtt{NotO_{n}^k}$.

By Prime Number Theorem\footnote{H.~G. Diamond.
\newblock {Elementary methods in the study of the distribution of prime
  numbers}.
\newblock {\em Bulletin of The American Mathematical Society}, 7:553--589,
  1982.
}, we know that $p_r= O(r\ln r)$. And so $p_i \leq O(r\ln r)$ for $1\leq i \leq r$ since $p_i \leq p_r$. This means the width of  $i$-th part is  $p_i\leq O(r\ln r)$.
Since $r\leq \log k$, the width of $P_n$ is $d\leq r \cdot O(r\ln r) =O(\log^2 k \ln\log k)$.

\section{The Proof of Lemma \ref{pneq1}}\label{app-pneq1}

We construct the NOBDD $P_n$ with the natural order. We use the fingerprinting method given in \cite{Fre79}. 
We will use the same notation given in the proof of Lemma \ref{peq1}.
Let $p_1,\dots,p_z$ be the first $z$ prime numbers satisfying that 
\[
p_1 p_2 \cdots p_z>2^{k/4},
\] where $z$ is the minimal value. Note that $r\leq k/4$ since $p_z\geq 2$.

The NOBDD $P_n$ consist of $z$ parallel parts, each of them corresponds to one of $p$ from $\{p_1,\dots,p_z\}$. 
Let us consider input $\nu\in\{0,1\}^n$ and let $\nu^i=(\nu_1,\dots, \nu_i)$ for $1\leq i \leq k$.

In the first step, $P_n$ nondetermenistically picks a $p$. Then, this branch computes $r(\nu)=bin(\alpha(\nu))-bin(\beta(\nu))\mod p)$. Here $bin(\alpha(\nu))$ is the binary represantation of $\alpha(\nu)$.

 For computing $r(\nu)$ in $i$-th step, we should know three numbers: the value of $r(\nu^i)=bin(\alpha(\nu^i))-bin(\beta(\nu^i))\mod p$, the length of $\alpha(\nu^i)$, and the length of $\beta(\nu^i)$. Note that $r(\nu^i)\in\{0,\dots, p-1\}$ and $|\alpha(\nu^i)|,|\beta(\nu^i)|\leq k/2$. If, in last step, $r(\nu^i)$ is not zero or  $|\alpha(\nu)|\neq|\beta(\nu)|$, then this part answers $1$ and $0$ otherwise. We need $p\cdot k^2/4$ nodes for each value of triple $\Big(r(\nu^i), |\alpha(\nu^i)|,|\beta(\nu^i)|\Big)$ and two times more nodes in order to check the value of odd bits for knowing whether the following bit belongs to $\alpha(\nu)$ or to $\beta(\nu)$. This means that the width of this branch is $p\cdot k^2/2$.

By the Chinese Remainder Theorem\footnote{K.~Ireland and M.~Rosen.
\newblock {A Classical Introduction to Modern Number Theory, 2nd ed.}
\newblock {\em Springer-Verlag}, pages 34--38, 1990.
}, if $bin(\alpha(\nu))\neq bin(\beta(\nu))$, then there is at least  one $p_i$ such that $bin(\alpha(\nu))\neq bin(\beta(\nu))\mod p_i$. That is, at least one branch gives the answer of $1$.
If the lengths are different, then each branch gives the answer of $1$.
But if $\alpha(\nu)= \beta(\nu)$, then, for all $p_i$, we have $bin(\alpha(\nu))= bin(\beta(\nu))\mod p_i$ and $|\alpha(\nu)|= |\beta(\nu)|$. 
That is, each branch gives the answer of $0$.
Hence $P_n$ computes $\mathtt{NotEQS^k_{n}}$.

By Prime Number Theorem\footnote{H.~G. Diamond.
\newblock {Elementary methods in the study of the distribution of prime
  numbers}.
\newblock {\em Bulletin of The American Mathematical Society}, 7:553--589,
  1982.
}, we know that $p_z= O(z\ln z)$ and so $p_i \leq O(z\ln z)$ for $1\leq i \leq z$ since $p_i \leq p_z$. This means that the width of  $i$-th part is  $p_i\cdot k^2/2\leq O(z\cdot k^2\ln z)$.
Since $z\leq k/4$, the width of $P_n$ is $d\leq z \cdot O(z\cdot k^2\ln z)  =O(k^4 \log k)$.

\section{The Proof of Lemma \ref{peq1}}\label{app-peq1}
We construct such OBDD $P_n$ which uses the natural order of variables.
Let $\nu\in\{0,1\}^n$ be the input. The main idea is to remember the bits from $\alpha(\nu)$ and $\beta(\nu)$ which have not been compared yet. Suppose that $P_n$ has already read $j$ bits of $\alpha(\nu)$ and $l$ bits of $\beta(\nu)$ at the $2i$-th level.  Each node of the level is associated with the value of string $c=(c_1,\dots,c_r)$ of bits which are not compared yet. For example, if $j>l$, then $c=(\alpha_{l+1}(\nu),\dots,\alpha_j(\nu))$; if $j<l$, then $c=(\beta_{j+1}(\nu),\dots,\beta_l(\nu))$; and, if $j=l$, then $c$ is empty string. 

Note that $c$ always contains the bits either from  $\alpha(\nu)$ or  from $\beta(\nu)$ but never both.
If some nodes $c$ contain the bits from $\alpha(\nu)$ and $P_n$ reads such a bit from $\alpha(\nu)$, then it adds this bit to $c$. Otherwise this bit belongs to $\beta(\nu)$ and $P_n$ compares it with the first bit of $c$. If both bits are equal, then $P_n$ removes this bit from $c$ and rejects the input otherwise. 

More formally, let us consider $2i$-th level. It contains four groups of nodes.  First of them are the nodes associated with $c$ that contains the bits from $\alpha(\nu)$. The second ones are the nodes for $c$ which contains bits from $\beta(\nu)$. The third one contains  only one ``equals'' node for empty $c$, and, the fourth one contains only one ``rejecting'' node.

Let $v_c$ be a node from the first group and $c=(c_1,\dots,c_r)$ be the string associated with $v_c$ such that $c$ contains the bits from $\alpha(\nu)$. At the $2i$-th level, $P_n$ reads a marker bit $\nu_{2i+1}$. If $\nu_{2i+1}=0$, that means the next value bit $\nu_{2i+2}$ belongs to $\alpha(\nu)$ and then  $P_n$ stores bit $\nu_{2i+2}$ at the level $2i+1$  and goes to the node corresponding to $c'=(c_1,\dots,c_r,\nu_{2i+2})$. Otherwise, the next value bit $\nu_{2i+2}$ belongs to $\beta(\nu)$ and $P_n$ compares it with $c_1$. If these bits are the same, then $P_n$ goes to the node corresponding to $c''=(c_2,\dots,c_r)$ or to ``equals'' node if $r=1$, and goes to ``rejecting'' node if the bits are different. If $c$ is empty, then $P_n$ goes to the node of the first or the second group associated to $c=(\nu_{2i+2})$ which contains bit from $\alpha$ or from $\beta$. 
If length of $c$ is greater than $k/4$, then $P_n$ goes to ``rejecting'' node. 

For the second group of nodes, $P_n$ works by the same way but the string $c$ stores the bits from $\beta(\nu)$.
$P_n$ gives the answer $1$ iff it reaches the ``equals'' node at the last level.

Now we compute the width of $P_n$.  At the $2i$-th level, the first two groups of nodes contain nodes for each value of $c$ both for $\alpha(\nu)$ and for $\beta(\nu)$. The third and forth groups contain ``equal'' and ``rejecting'' nodes, respectively. The width of such level is

\[d= (2+4+\dots +2^{k/4}) + (2+4+\dots +2^{k/4}) + 2=\]\[=(2^{k/4+1}-2) + (2^{k/4+1}-2) + 2=4\cdot 2^{k/4}-2.\]

The level $2i+1$ have twice more nodes for the first three groups since $P_n$ needs to remember the value of marker bit, which indicates the next bit belongs to $\alpha(\nu)$ or $\beta(\nu)$. So, the widths of these levels are $8\cdot 2^{k/4}-5$ and so the width of $P_n$ is $8\cdot 2^{k/4}-5$.

Note that OBDD is particular case of NOBDD and the same result is followed for NOBDDs.

\section{The Proof of Lemma \ref{peq2}}
\label{app-peq2}

Let $P_n$ be an OBDD with order $\pi=(j_1,\dots,j_n)$ that computes $ \mathtt{EQS^k_{n}} $ and $l$ be the number of level such that the $l$-th level of $P_n$ has already read exactly $k/4$ value bits from $X_K=\{x_1,\dots,x_k\}$. 

Let us consider partition of variables $(\{x_{j_1},\dots,x_{j_l}\},\{x_{j_{l+1}},\dots,x_{j_n}\})=(X_A,X_B)$ and input $\nu=(\sigma,\gamma)$ with respect to this partition. Let set $X_V$ is all value bits from $X_A\cap X_K$ and set $X_M=(X_A\cap X_K)\backslash X_V$ is all marker bits from $X_A\cap X_K$.

Let set $\Sigma=\{\sigma\in\{0,1\}^l:$ marker bits for variables from $X_M$ fixed such that value bits for variables from $X_V$ belongs to $\alpha(\nu)$ and other value bit from $X_K\backslash X_V$ belongs to $\beta(\nu)\}$. Note that $|\Sigma|=2^{k/4}$. 

By the same way as in Lemma \ref{peq1s} we can show that each $\sigma\in \Sigma$ reaches different nodes on  level $l$, and therefore the width of $l$-th level at least $2^{k/4}$.

\end{document}